\documentclass[a4paper,12pt]{article}
\usepackage[margin=2cm]{geometry}

\usepackage[utf8]{inputenc}

\usepackage{float}
\usepackage{amssymb}
\usepackage{amsmath}
\usepackage{amsfonts}
\usepackage{graphicx}
\usepackage{algorithm}
\usepackage{algorithmic}
\usepackage[dvipsnames]{xcolor}
\usepackage[caption = false]{subfig}
\usepackage{tikz-cd}

\renewcommand{\bar}{\overline}

\newcommand{\bey}{\begin{eqnarray}}
\newcommand{\pslash}{\not{\hbox{\kern-2.3pt $p$}}}
\newcommand{\pdslash}{\not{\hbox{\kern-2pt $\partial$}}}
\newcommand{\eey}{\end{eqnarray}}
\newtheorem{theorem}{Theorem}
\newtheorem{proposition}{Proposition}
\newenvironment{proof}[1][Proof]{\noindent\textbf{#1.} }{\ \rule{0.5em}{0.5em}}
\newtheorem{lemma}{Lemma}
\newtheorem{definition}{Definition}

\newcommand{\ket}[1]{\vert #1 \rangle}
\newcommand{\bra}[1]{\langle #1 \vert}
\newcommand{\innerproduct}[2]{\langle #1 \vert #2 \rangle}

\newcommand{\edittext}{\color{black}}

\begin{document}
    
\begin{titlepage}
\vskip 2cm
\begin{center}
{\Large Quantum classification and search algorithms using spinorial representations.

\vskip 10pt
{ Lauro Mascarenhas $^{\dagger}$ \footnote{{\tt  laurodejesusmascarenhas@gmail.com}}, Vinicius N. A. Lula-Rocha $^{\ddagger}$\footnote{{\tt viniciusnonato@gmail.com}} \\ 
Marco A. S. Trindade $^{\dagger}$ \footnote{{\tt matrindade@uneb.br}}} \\}

\vskip 5pt
{\sl $^{\dagger}$ Colegiado de Física, Departamento de Ciências Exatas e da Terra, \\
Universidade do Estado da Bahia\\
Rua Silveira Martins, Cabula, 41150000, Salvador, Bahia, Brazil\\} 
\vskip 5pt
{\sl $^{\ddagger}$ Departamento de Física, Instituto de Ciências Naturais \\
Universidade Federal de Lavras\\
37200-900 Lavras, Minas Gerais, Brazil\\}

\vskip 5pt
\end{center}

\begin{abstract}
We propose an algebraic formulation for two distinct quantum algorithms: a quantum classification algorithm and a quantum search algorithm with a non-uniform initial distribution, both based on Clifford algebras and spinorial representations. In the classification algorithm, we exploit properties of spinorial representations to construct orthogonal quantum states associated with different classes, allowing the identification of an item’s class through the evaluation of expectation values of operators derived from the generators of the Clifford algebra. In the quantum search algorithm, we consider a database with prior information in which the oracle is implemented directly using generators of the Clifford algebra, simplifying its realization. The proposed approach provides a unified algebraic description for both algorithms, employing spinorial representations in the construction of quantum states and operators. Computational implementations are presented.
\end{abstract}

\bigskip

{\it Keywords:} Quantum classification algorithm, Quantum classification algorithm, Clifford algebras, spinorial representations.

\end{titlepage}

\newpage
\setcounter{page}{1}

\section{Introduction}

Quantum search algorithms are among the most relevant and well-established algorithms in quantum computation, both from a theoretical and a practical perspective \cite{Grover, Brassard, Nielsen, Scherer}. Beyond their intrinsic significance as paradigmatic examples of quantum speedup in relation with the classical versions, these algorithms frequently appear as essential subroutines in more elaborate quantum protocols, enabling performance improvements in a wide range of computational tasks. In particular, Grover’s search algorithm and its generalizations provide a quadratic speedup for unstructured search problems and constitute  fundamental building blocks for amplitude amplification techniques.

In the context of quantum machine learning, quantum search algorithms have been investigated as a means to accelerate classical learning and optimization procedures \cite{Nielsen, Grover2}. For instance, Grover’s algorithm and its variants have been employed to enhance classical clustering algorithms by quantizing specific subroutines, thereby achieving faster convergence in unsupervised learning tasks \cite{Aimeur}. The central idea underlying these approaches is to obtain speedups by selectively replacing classical search or sampling steps by their quantum counterparts. Moreover, amplitude amplification techniques have also been proposed for the training of quantum perceptrons, leading to improved performance in learning models that rely on iterative optimization procedures \cite{Kapoor, Schuld}. 

In a related line of research, some authors have investigated the use of Clifford algebras as an algebraic framework for the analysis of Grover’s search algorithm \cite{Gregoric, Alves}. In Ref.~\cite{Gregoric}, the authors employ Clifford algebra techniques to derive a generalized version of Grover’s algorithm, which naturally reduces to the standard formulation for a particular choice of the initial state. This algebraic approach provides a compact and geometrically transparent description of the algorithm’s dynamics, allowing the search process to be interpreted in terms of rotations in suitably defined subspaces. From a similar perspective, Ref.~\cite{Alves} demonstrates that the computation of algorithmic complexity can be significantly simplified by recasting the analysis in terms of Clifford algebraic structures. 

Clifford algebras and spinors play a fundamental role in several areas of physics, providing a unifying algebraic framework for the description of fermionic degrees of freedom and symmetry structures \cite{Doran, Lounesto, Hestenes, Vaz}. In relativistic quantum mechanics, for instance, the Dirac matrices furnish a concrete representation of the Clifford algebra $Cl(1,3)$, forming the algebraic backbone of the Dirac equation and the theory of relativistic spin-$\tfrac{1}{2}$ particles. In supersymmetry and in M-theory, Clifford algebras and their associated spinorial representations are ubiquitous, appearing naturally in the construction of supersymmetry algebras, supercharges, and higher-dimensional field theories \cite{Traubenberg, Anastasiou, Toppan, Trindade1}. This is particularly evident in eleven-dimensional supergravity, where the supercharge is represented by Majorana spinors. Beyond high-energy physics, these algebraic structures have also found applications in quantum information science. In the context of quantum neural networks and quantum Boltzmann machines, it has been shown that Clifford algebras can be employed to encode and extract geometric information inherent in the data, providing a useful tool for the development of quantum learning models \cite{Trindade2, Trindade3}.

In this work, we explore the use of Clifford algebras as a mathematical framework for the construction of a quantum classification algorithm. Our approach relies on spinorial representations, which allow the systematic construction of orthogonal quantum states and operators directly from generators of the Clifford algebra. The formulation developed in this context can also be applied to quantum search problems in which the initial state of the system is not fixed a priori, but instead is constructed from spinorial representations of the group $\mathrm{Spin}(2n)$. This setting is particularly relevant in light of Ref.~\cite{Biham}, where a version of Grover’s search algorithm with arbitrary initial amplitudes was derived. In such scenarios, the initial distribution of marked and unmarked states may corresponds to the output of a preceding quantum process, rather than to a uniform superposition. Consequently, the system may start its evolution in a more general quantum state, making it necessary to consider search procedures that go beyond the standard assumptions of the original Grover algorithm.

Our work is organized as follows. In Section~2, we introduce the algebraic formulation underlying our approach, which is based on Clifford algebras and spinorial representations, and present the mathematical results that support the construction of the proposed quantum states and operators. Section~3 is devoted to the description of the quantum algorithms introduced in this work, including both the classification algorithm and the quantum search algorithm with prior information, together with illustrative examples and discussions of their operational principles. In Section~4, we present our conclusions and discuss perspectives and directions for future research. Finally, in the Appendix, we review several foundational results on representations of Clifford algebras and spinorial representations.
\section{Algebraic formulation}

In this section, we present the mathematical results that provide the basis for the construction of our algorithms. Our formulation is based on spinorial representations of $\mathrm{Spin}(2n)$. The following lemma shows how to construct orthogonal quantum  states in the context of representations of Clifford algebras.

\begin{lemma} \label{lem1}
Given a representation of a Clifford algebra $Cl(2n)$ defined by:
\begin{equation} \label{eq_def_gerador_clifford}
\begin{matrix}
\Gamma_j = \sigma_1^{\otimes (j-1)} \otimes \sigma_2 \otimes I^{\otimes (n - j )} \\
\Gamma_{n+j} = \sigma_1^{\otimes (j-1)} \otimes \sigma_3 \otimes I^{\otimes (n - j )}
\end{matrix}
\end{equation}
where $\sigma_1, \sigma_2, \sigma_3$ are Pauli matrices, $I$ is the identity matrix and $1 \leq j \leq n$ \cite{Gilbert}, there is an eigenvector $\ket{\Gamma_j}$ of $\Gamma_{j}$ orthogonal to the vector $\Gamma_i \Gamma_j \ket{\Gamma_j}$, with $i\neq j$ . The eigenvectors of $\Gamma_j$ and $\Gamma_{n+j}$ are given explicitly by
{\edittext
\begin{equation} \label{2}
\left( \frac{1}{\sqrt{2}} \right)^{j-1}
\begin{pmatrix}
1 \\ \pm 1
\end{pmatrix}^{\otimes j-1}
\otimes
\begin{pmatrix}
1 \\ \pm i
\end{pmatrix}
\otimes
\begin{pmatrix}1 \\ 0\end{pmatrix}^{\otimes n-j}
\end{equation}
and
\begin{equation} \label{3}
\left( \frac{1}{\sqrt{2}} \right)^{j-1}
\begin{pmatrix}
1 \\ \pm 1
\end{pmatrix}^{\otimes j-1}
\otimes
\begin{pmatrix}1 \\ 0\end{pmatrix}^{\otimes n-j+1}
\end{equation}
}
respectively, with eigenvalues $\pm 1$.
\end{lemma}

\begin{proof}
Let $\lambda_j \neq 0$ be an eigenvalue associated with the operator $\Gamma_j$ acting on an eigenvector $\ket{\Gamma_j}$. Then, for $i \neq j$, we have
\begin{equation*}
\begin{split}
\Gamma_i \Gamma_j \ket{\Gamma_j}
&= -\Gamma_j \Gamma_i \ket{\Gamma_j} \\
&= -\Gamma_j \Gamma_i \left( \frac{1}{\lambda_j} \Gamma_j \ket{\Gamma_j} \right) \\
&= -\frac{1}{\lambda_j} \Gamma_j \Gamma_i \Gamma_j \ket{\Gamma_j}.
\end{split}
\end{equation*}
Consequently,
\begin{equation}
\Gamma_j \Gamma_i \Gamma_j \ket{\Gamma_j}
= -\lambda_j \Gamma_i \Gamma_j \ket{\Gamma_j},
\end{equation}
that is, $\Gamma_i \Gamma_j \ket{\Gamma_j}$ is an eigenvector of $\Gamma_j$ with eigenvalue $-\lambda_j$.

Let $\ket{v_1}$ and $\ket{v_2}$ be two eigenvectors of $\Gamma$ (an arbitrary generator of Clifford algebra $Cl(2n)$) with distinct eigenvalues. Then
\begin{equation*}
\begin{split}
\bra{v_2} \Gamma \ket{v_1}
&= \bra{v_2} \Gamma^\dagger \ket{v_1} \\
&= (\bra{v_2} \Gamma) \ket{v_1}, \\
\lambda_1 \innerproduct{v_2}{v_1}
&= \lambda_2 \innerproduct{v_2}{v_1},
\end{split}
\end{equation*}
since the operators $\Gamma$ are Hermitian. It follows that
\begin{equation}
\lambda_1 \neq \lambda_2 \Rightarrow \innerproduct{v_2}{v_1} = 0.
\end{equation}

{\edittext
Therefore, the eigenvectors are orthogonal. The eigenvectors $\ket{\Gamma_j}$ and $\ket{\Gamma_{n+j}}$ are eigenvectors of $\Gamma_j$ and $\Gamma_{n+j}$, respectively, and are constructed from eigenvectors of the Pauli matrices. Hence, the eigenvectors can be written as
\begin{equation}
\ket{\Gamma_j}
=
\ket{\sigma_1}^{\otimes j-1}
\otimes
\ket{\sigma_2}
\otimes
\ket{\mathbb{I}}^{\otimes n-j+1},
\end{equation}
\begin{equation}
\ket{\Gamma_{n+j}}
=
\ket{\sigma_1}^{\otimes j-1}
\otimes
\ket{\sigma_3}
\otimes
\ket{\mathbb{I}}^{\otimes n-j+1}.
\end{equation}
Thus,
\begin{equation}
\Gamma_j \ket{\Gamma_j}
=
\lambda_{\sigma_1} \cdots \lambda_{\sigma_1}
\lambda_{\sigma_2}
\lambda_{\mathbb{I}} \cdots \lambda_{\mathbb{I}}
\ket{\Gamma_j}.
\end{equation}
where $\lambda_{\sigma_1}$ and $\lambda_{\mathbb{I}}=1$ are the eigenvalues of Pauli matrices and of the identity matrix, respectively. It is straightforward to verify that the eigenvalues are $\pm 1$. Indeed,
\begin{eqnarray}
\Gamma_j \ket{\Gamma_j} &=& \lambda_j \ket{\Gamma_j}, \nonumber \\
\Gamma_j^2 \ket{\Gamma_j} &=& \lambda_j^2 \ket{\Gamma_j}, \nonumber \\
\ket{\Gamma_j} &=& \lambda_j^2 \ket{\Gamma_j}. \nonumber
\end{eqnarray}
This implies $\lambda_j = \pm 1$, which provides the justification for eqs. (\ref{2}) and (\ref{3}). In summary,
\begin{equation}
\bra{\Gamma_j} \Gamma_i \Gamma_j \ket{\Gamma_j} = \delta_{ij}.
\end{equation}
}
\end{proof}

Next, we introduce the concept of a class, which is fundamental for the construction of our first algorithm.

\begin{definition}
Let $\mathcal H$ be a Hilbert space and let $\{O_k\}_{k=1}^{2n}$ be a set of linear operators defined on
$\mathcal H$. We define the class type-I, $\mathcal C_i$, as the set of states
$\vert \psi\rangle \in \mathcal H$ such that there exists a nonzero constant $c_i$ satisfying
\begin{eqnarray}
\langle \psi \vert O_k \vert \psi \rangle
=
\delta_{ik}\, c_i,
\qquad k=1,\dots,2n,
\end{eqnarray}
where $\delta_{ik}$ denotes the Kronecker delta. Equivalently,
\begin{eqnarray}
\vert \psi\rangle \in \mathcal C_i
\iff
\begin{cases}
\langle O_i \rangle_\psi = c_i \neq 0,\\
\langle O_k \rangle_\psi = 0,
\quad k\neq i.
\end{cases}
\end{eqnarray}
\end{definition}

\begin{theorem}
There are states $\vert \psi \rangle$ and operators $\Gamma_{i}$ ( generators of Cl(2n)) such that the state $\vert \psi \rangle \in C_i$.
\end{theorem}
\begin{proof}
Let $|\psi \rangle$ be the state obtained of $R_{ij} \in Spin(2n)$ as
\begin{eqnarray}
|\psi \rangle &=& R_{ij}|\psi \rangle \nonumber \\
&=&\exp\left(\theta_{ij}\Gamma_i\Gamma_j\right) \vert \Gamma_j \rangle \nonumber \\
&=&\cos \theta_{ij}\ket{\Gamma_j}+\sin \theta_{ij}\Gamma_{i}\ket{\Gamma_j}
\end{eqnarray}
Therefore
\begin{eqnarray}
\langle \psi \vert \Gamma_k \vert \psi \rangle
&=&
\left(
\cos\theta_{ij}\langle \Gamma_j \vert
+
\sin\theta_{ij}\langle \Gamma_j \vert \Gamma_i
\right)
\Gamma_k
\left(
\cos\theta_{ij}\vert \Gamma_j \rangle
+
\sin\theta_{ij}\Gamma_i \vert \Gamma_j \rangle
\right)
\nonumber\\
&=&
\cos^2\theta_{ij}
\langle \Gamma_j \vert \Gamma_k \vert \Gamma_j \rangle
+
\cos\theta_{ij}\sin\theta_{ij}
\langle \Gamma_j \vert \Gamma_k\Gamma_i \vert \Gamma_j \rangle
\nonumber\\
&+&
\cos\theta_{ij}\sin\theta_{ij}
\langle \Gamma_j \vert \Gamma_i\Gamma_k \vert \Gamma_j \rangle
+
\sin^2\theta_{ij}
\langle \Gamma_j \vert \Gamma_i\Gamma_k\Gamma_i \vert \Gamma_j \rangle \nonumber \\
&=&2 \sin\theta_{ij}\cos\theta_{ij} \delta_{ik},
\end{eqnarray}
using the Lemma \ref{lem1} and $\{\Gamma_i, \Gamma_j\}=2\delta_{ij}$
\end{proof}
\begin{definition}
A class type-II is defined as an $m$-tuple
$(\mathrm{sgn}(M_{j_1}), \ldots, \mathrm{sgn}(M_{j_m}))$, if there are states $\vert \psi_{i_l j_l} \rangle$ and a Hermitian operator $M_{j_k}$ such that
\begin{eqnarray}
\langle M_{j_k} \rangle
=
\langle \psi_{i_1 j_1}, \ldots, \psi_{i_m j_m}
\vert
M_{j_k}
\vert
\psi_{i_1 j_1}, \ldots, \psi_{i_m j_m}
\rangle
\in [-1,1],
\end{eqnarray}
and $\mathrm{sgn}(M_{j_k}) = 1$ if $\langle M_{j_k} \rangle \geq 0$ and
$\mathrm{sgn}(M_{j_k}) = -1$ if $\langle M_{j_k} \rangle < 0$. We define $\vert \psi_{i_1 j_1}, \ldots, \psi_{i_m j_m}
\rangle = \vert \psi_{i_1 j_1} \rangle \otimes \ldots \otimes \vert \psi_{i_m j_m} \rangle$.

\end{definition}

We consider now the tensor product of Clifford algebras $Cl(2n)^{\otimes n}$ and consequently  $\mathrm{Spin}(2n)^{\otimes n}$. 
\begin{theorem}
There is an operator $R_{i_{1},\cdots,i_{m},j_{1},\cdots j_{m}} \in \mathrm{Spin}(2n)^{\otimes n}$ and states
$\vert \psi_{i_{1},j_{1},\cdots,i_{n}j_{n}} \rangle$ given by
\begin{eqnarray}
\vert \psi_{i_{1},j_{1},\cdots,i_{m}j_{m}} \rangle
=
R_{i_{1}j_{1},\cdots,i_{m}j_{m}}
(\vert \Gamma_{j_{1}} \rangle \otimes \cdots \otimes \vert \Gamma_{j_{k}}\rangle \otimes \cdots \otimes \vert \Gamma_{j_{m}} \rangle),
\end{eqnarray}
such that $\operatorname{sgn}(M_{j_{k}})$ defines the elements of a class for
$M_{j_{k}} = I \otimes \cdots \otimes \Gamma_{j_{k}} \otimes \cdots \otimes I$.
\end{theorem}

\begin{proof}
Define
\begin{eqnarray}
R_{i_{1},\cdots,i_{m},j_{1},\cdots j_{m}}
&=&
\exp\!\Bigl[
\theta_{i_{1}j_{1}}(\Gamma_{i_{1}}\Gamma_{j_{1}} \otimes \cdots \otimes I)
+ \cdots
+ \theta_{i_{k}j_{k}} (I \otimes \cdots \otimes \Gamma_{i_{k}}\Gamma_{j_{k}} \otimes \cdots \otimes I)
\nonumber \\
&+&
\cdots
+ \theta_{i_{m}j_{m}}(I \otimes \cdots \otimes \Gamma_{i_{m}}\Gamma_{j_{m}})
\Bigr].
\end{eqnarray}

We can rewrite $R_{i_{1},\cdots,i_{n},j_{1},\cdots j_{n}}$ as
\begin{eqnarray}
R_{i_{1},\cdots,i_{m},j_{1},\cdots j_{m}}
=
\exp(\theta_{i_{1}j_{1}}\Gamma_{i_{1}}\Gamma_{j_{1}})
\otimes \cdots \otimes
\exp(\theta_{i_{k}j_{k}}\Gamma_{i_{k}}\Gamma_{j_{k}})
\otimes \cdots \otimes
\exp(\theta_{i_{n}j_{m}}\Gamma_{i_{n}}\Gamma_{j_{m}}).
\end{eqnarray}

Consequently,
\begin{eqnarray}
R_{i_{1},\cdots,i_{m},j_{1},\cdots,j_{m}}
\Bigl(
\vert \Gamma_{j_{1}} \rangle \otimes \cdots \otimes \vert \Gamma_{j_{k}} \rangle \otimes \cdots \otimes \vert \Gamma_{j_{n}} \rangle
\Bigr)
&=&
\Bigl[
\cos(\theta_{i_{1}j_{1}})\, \vert \Gamma_{j_{1}} \rangle
+
\sin(\theta_{i_{1}j_{1}})\, \Gamma_{i_{1}} \vert \Gamma_{j_{1}} \rangle
\Bigr] \otimes \nonumber \\
&\cdots&
\otimes
\Bigl[
\cos(\theta_{i_{k}j_{k}})\, \vert \Gamma_{j_{k}} \rangle
+
\sin(\theta_{i_{k}j_{k}})\, \Gamma_{i_{k}} \vert \Gamma_{j_{k}} \rangle
\Bigr] \otimes \nonumber \\
&\cdots&
\otimes
\Bigl[
\cos(\theta_{i_{m}j_{m}})\, \vert \Gamma_{j_{n}} \rangle
+
\sin(\theta_{i_{m}j_{m}})\, \Gamma_{i_{m}} \vert \Gamma_{j_{m}} \rangle
\Bigr]. \nonumber
\end{eqnarray}

Therefore,
\begin{eqnarray}
\langle M_{jk} \rangle
&=&
\langle \psi_{i_{1}j_{1}},\cdots, \psi_{i_{k}j_{k}},\cdots, \psi_{i_{m}j_{m}}
\vert M_{jk} \vert
\psi_{i_{1}j_{1}},\cdots, \psi_{i_{k}j_{k}},\cdots, \psi_{i_{m}j_{m}} \rangle
\nonumber \\
&=&
\langle I \otimes \cdots \otimes \Gamma_{j_{k}} \otimes \cdots \otimes I \rangle
\nonumber \\
&=&
\Bigl(
\cos \theta_{i_{k}j_{k}} \langle \Gamma_{j_{k}} \vert
+
\sin \theta_{i_{k}j_{k}} \langle \Gamma_{j_{k}} \vert \Gamma_{i_{k}}
\Bigr)
\Gamma_{j_{k}}
\Bigl(
\cos \theta_{i_{k}j_{k}} \vert \Gamma_{j_{k}} \rangle
+
\sin \theta_{i_{k}j_{k}} \Gamma_{i_{k}} \vert \Gamma_{j_{k}} \rangle
\Bigr)
\nonumber \\
&=&
\cos^{2}\theta_{i_{k}j_{k}} - \sin^{2}\theta_{i_{k}j_{k}},
\end{eqnarray}
using the orthogonality, by Lemma~1. Therefore, $\langle M_{jk} \rangle \in [-1,1]$, and the theorem is proved.
\end{proof}

We could also perform a class decomposition using the chiral decomposition of the tensor product of irreducible representations. The existence of these classes can be better understood mathematically with the following theorem. 

\begin{theorem}
Let $V=\mathbb{R}^{2n}$ be an oriented Euclidean vector space and let
\begin{equation}
\rho : \mathrm{Spin}(2n) \longrightarrow \mathrm{End}(S)
\end{equation}
be the complex spinorial representation, where $S$ is the spinor space. For any integer $m\ge 1$, consider the tensor product representation
\begin{equation}
\rho^{\otimes m} : \mathrm{Spin}(2n)^{\otimes m} \longrightarrow \mathrm{End}(S^{\otimes m}),
\end{equation}

Then there exists a canonical decomposition
\begin{equation}
S^{\otimes m}
=
\bigoplus_{\varepsilon_1,\ldots,\varepsilon_m\in\{+,-\}}
S^{\varepsilon_1}\otimes \cdots \otimes S^{\varepsilon_m},
\end{equation}
where each $S^{\varepsilon_k}\subset S$ is defined by the spectrum of the chirality operator, and every summand
\begin{equation}
S^{\varepsilon_1}\otimes \cdots \otimes S^{\varepsilon_m}
\subset S^{\otimes m}
\end{equation}
is invariant under the action of $\mathrm{Spin}(2n)$ via $\rho^{\otimes m}$. The symbols $+$ and $-$ correspond to the subspaces with eigenvalues $+1$ and $-1$, respectively.
\end{theorem}

\begin{proof}
We consider the tensor product space $S^{\otimes m}$ endowed with the diagonal action of $\mathrm{Spin}(2n)$,
\begin{equation}
\rho^{\otimes m}(g) = \rho(g)\otimes \cdots \otimes \rho(g)
\in \mathrm{End}(S^{\otimes m}).
\end{equation}

Let $\{e_1,\ldots,e_{2n}\}$ be an oriented orthonormal basis of $V$ and define
\begin{equation}
\Gamma_i := \rho(e_i) \in \mathrm{End}(S),
\end{equation}
which satisfy the Clifford relations
\begin{equation}
\Gamma_i \Gamma_j + \Gamma_j \Gamma_i = 2\delta_{ij} I.
\end{equation}
We use the same symbol $\rho$ to denote both the spinorial representation of $\mathrm{Spin}(2n)$ and its canonical extension to the Clifford algebra, the intended meaning being clear from the domain.  The chirality operator is defined by
\begin{equation}
\Gamma_* := i^{\,n}\,\Gamma_1 \Gamma_2 \cdots \Gamma_{2n}.
\end{equation}
Since $V$ is oriented, this operator is well defined and satisfies
\begin{equation}
\Gamma_*^2 = I,
\qquad
\Gamma_*^\dagger = \Gamma_*.
\end{equation}
Hence, $\Gamma_*$ is diagonalizable with eigenvalues $\pm 1$, and we define
\begin{equation}
S^\pm := \ker(\Gamma_* \mp I).
\end{equation}
On the tensor product space $S^{\otimes m}$, define the operators
\begin{equation}
\Gamma_*^{(k)} :=
\underbrace{I \otimes \cdots \otimes I \otimes}_{k-1} \Gamma_* \underbrace{\otimes I \otimes \cdots \otimes I}_{n-k+1}
\in \mathrm{End}(S^{\otimes m}),
\end{equation}
where $\Gamma_*$ acts on the $k$-th factor. These operators commute and are self-adjoint. For $\varepsilon_1,\ldots,\varepsilon_m \in \{+,-\}$, we have that
\begin{equation}
S^{\varepsilon_1}\otimes \cdots \otimes S^{\varepsilon_m}
=
\bigcap_{k=1}^m
\ker\bigl(\Gamma_*^{(k)} - \varepsilon_k I\bigr).
\end{equation}
By the spectral theorem applied to the commuting family
$\{\Gamma_*^{(k)}\}_{k=1}^m$, we obtain the orthogonal decomposition
\begin{equation}
S^{\otimes m}
=
\bigoplus_{\varepsilon_1,\ldots,\varepsilon_m\in\{+,-\}}
S^{\varepsilon_1}\otimes \cdots \otimes S^{\varepsilon_m}.
\end{equation}
Finally, the Lie algebra $\mathfrak{spin}(2n)$ acts on $S$ via operators of the form (see Appendix)
\begin{eqnarray}
\Gamma_i \Gamma_j ,
\end{eqnarray}
which commute with $\Gamma^*_k$. Therefore, for the representation $\rho$ extended to the algebra, we have
\begin{eqnarray}
[\Gamma^*_k, \rho(X)] = 0, \quad \forall X \in \mathfrak{spin}(2n). \label{eq:lie_comm}
\end{eqnarray}
Since the group $\mathrm{Spin}(2n)$ is generated by exponentials of elements of the algebra, i.e.,
\begin{eqnarray}
g' = \exp(X), \quad X \in \mathfrak{spin}(2n),
\end{eqnarray}
and $\rho$ is a linear representation, it follows that
\begin{eqnarray}
\rho(g') = \rho(\exp(X)) = \exp(\rho(X)).
\end{eqnarray}
Because $\Gamma^*_k$ commutes with $\rho(X)$, it also commutes with $\exp(\rho(X))$, giving
\begin{eqnarray}
[\Gamma^*_k, \rho(g')] = 0, \quad \forall g' \in \mathrm{Spin}(2n).
\end{eqnarray}
Thus, for any $\psi \in S^{\otimes m}$,
\begin{eqnarray}
(\Gamma^*_k \mp I)\rho(g')\psi = \rho(g')(\Gamma^*_k \mp I)\psi = 0,
\end{eqnarray}
showing that $\rho(g') \in \ker(\Gamma^*_k \mp I)$, which proves the claim.
\end{proof}

\noindent\textbf{Remark.}
Although the representation of the group $\mathrm{Spin}(2n)$ takes values in $\mathrm{GL}(S)$, its canonical extension to the Clifford algebra is naturally an algebra representation with values in $\mathrm{End}(S)$, since generic elements of the algebra are not necessarily invertible.

The following diagrams illustrate the fact that the chiral decomposition of the tensor product spinor space is compatible with the algebraic action of the Clifford algebra. By means of the canonical projectors onto the positive and negative chiral subspaces, the diagrams illustrate that the action of the total chirality operator intertwines with the projection onto each eigenspace. In particular, it shows that applying the chirality operator before projecting onto a given chiral sector, or first projecting and then applying the restricted action, leads to the same result up to the corresponding eigenvalue.

\[
\begin{tikzcd}
S^{\otimes m}
\arrow[r, "P_{+}"]
\arrow[d, "\Gamma^{(m)}_*"']
&
(S^{\otimes m})^{+}
\arrow[d, "+\,\mathrm{id}"]
\\
S^{\otimes m}
\arrow[r, "P_{+}"']
&
(S^{\otimes m})^{+}
\end{tikzcd}
\qquad
\begin{tikzcd}
S^{\otimes m}
\arrow[r, "P_{-}"]
\arrow[d, "\Gamma^{(m)}_*"']
&
(S^{\otimes m})^{-}
\arrow[d, "-\,\mathrm{id}"]
\\
S^{\otimes m}
\arrow[r, "P_{-}"']
&
(S^{\otimes m})^{-}
\end{tikzcd}
\]
where,
\[
P_{\pm} = \frac{1}{2}\bigl(1 \pm \Gamma^{(m)}_*\bigr).
\]
and $(S^{\otimes m})^{+}$ and $(S^{\otimes m})^{-}$ are the subspaces of $S^{\otimes m}$ with eigenvalues $+1$ and $-1$, respectively.

In order to illustrate our formulation, consider the Clifford algebra $Cl(4)$. 
Two of its generators can be written in terms of tensor products of Pauli matrices as
\begin{eqnarray}
\Gamma_1 &=& \sigma_2 \otimes I =
\begin{pmatrix}
0 & 0 & -i & 0 \\
0 & 0 & 0 & -i \\
i & 0 & 0 & 0 \\
0 & i & 0 & 0
\end{pmatrix},
\label{eq:Gamma1}
\end{eqnarray}
and
\begin{eqnarray}
\Gamma_2 &=& \sigma_1 \otimes \sigma_2 =
\begin{pmatrix}
0 & 0 & 0 & -i \\
0 & 0 & i & 0 \\
0 & -i & 0 & 0 \\
i & 0 & 0 & 0
\end{pmatrix}.
\label{eq:Gamma2}
\end{eqnarray}
The orthonormal spinorial states obtained from these generators are given by
\begin{eqnarray}
\ket{\Gamma_1}
&=& \frac{1}{\sqrt{2}}\left(\ket{00} + i\ket{10}\right)
= \frac{1}{\sqrt{2}}
\begin{pmatrix}
1 \\
0 \\
i \\
0
\end{pmatrix},
\label{eq:Gamma1_state}
\end{eqnarray}
and
\begin{eqnarray}
\Gamma_2 \ket{\Gamma_1}
&=& \frac{i}{\sqrt{2}}\left(\ket{01} + \ket{11}\right)
= \frac{1}{\sqrt{2}}
\begin{pmatrix}
0 \\
i \\
0 \\
i
\end{pmatrix}.
\label{eq:Gamma2Gamma1_state}
\end{eqnarray}
Using the generators $\Gamma_1, \Gamma_2$, one can then construct the general state
\begin{eqnarray}
\ket{\psi}
&=& \exp\left(\frac{\theta}{2} \Gamma_{2} \Gamma_{1}\right)\ket{\Gamma_{1}} \label{eq:psi_def} \\
&=& \cos(\theta/2)\ket{\Gamma_1}
+ \sin(\theta/2)\, \Gamma_2 \Gamma_1 \ket{\Gamma_1} \label{eq:psi_expansion} \\
&=& \frac{1}{\sqrt{2}}
\begin{pmatrix}
\cos(\theta/2) \\
i\sin(\theta/2) \\
i\cos(\theta/2) \\
i\sin(\theta/2)
\end{pmatrix}.
\label{eq:psi_components}
\end{eqnarray}

Considering that, in our proposal, the initial states do not need to be in a uniform superposition and that we may consider several subspaces of interest, we obtain a generalized version of the Grover operator. This will be useful for our second algorithm.

\begin{definition}
The generalized Grover operator is defined as
\begin{equation} \label{grover_g_equation}
\hat{G}
=
\left(
2\sum_{i=1}^{k} \ket{\psi_i}\bra{\psi_i}
-
\mathbb{I}
\right)
\prod_{i=1}^{k} \hat{O}_i .
\end{equation}
\end{definition}

The generalized form of $\ket{\psi_i}$ can be written as
\begin{equation}
\ket{\psi_i}
=
R_i(\theta)\, \ket{0}^{\otimes n},
\end{equation}
with $\langle \psi_i \vert \psi_j \rangle = \delta_{ij}$. The operator $R_i(\theta)$ is defined as
\begin{equation}
R_i(\theta)
=
\exp\left[
i \sum_{j_1 \cdots j_n}
\omega_{j_1 \cdots j_n}^{i}\,
\theta_{j_1 \cdots j_n}^{i}\,
\Gamma_{j_1}\cdots\Gamma_{j_n}
\right].
\end{equation}
We must show that this operator is unitary.

\begin{proposition}
The generalized Grover operator, given by Eq.~\eqref{grover_g_equation}, is unitary.
\end{proposition}

\begin{proof}

\begin{eqnarray}
\hat{G}\hat{G}^{\dagger}
=
\left(
2\sum_{i=1}^{k} \ket{\psi_i}\bra{\psi_i}
-
\mathbb{I}
\right)
\left(\prod_{i=1}^{k} \hat{O}_i\right)
\left[
\left(
2\sum_{j=1}^{k} \ket{\psi_j}\bra{\psi_j}
-
\mathbb{I}
\right)
\left(\prod_{j=1}^{k} \hat{O}_j\right)
\right]^{\dagger}. \nonumber 
\end{eqnarray}

\begin{eqnarray}
=
\left(
2\sum_{i=1}^{k} \ket{\psi_i}\bra{\psi_i}
-
\mathbb{I}
\right)
\left(\prod_{i=1}^{k} \hat{O}_i\right)
\left(\prod_{j=1}^{k} \hat{O}_j^{\dagger}\right)
\left(
2\sum_{j=1}^{k} \ket{\psi_j}\bra{\psi_j}
-
\mathbb{I}
\right).
\end{eqnarray}

Since the oracles are unitary operators, $\hat{O}_i \hat{O}_i^{\dagger} = \mathbb{I}$, and therefore
\begin{eqnarray}
\left(\prod_{i=1}^{k} \hat{O}_i\right)
\left(\prod_{j=1}^{k} \hat{O}_j^{\dagger}\right)
=
\mathbb{I}.
\end{eqnarray}

Thus,
\begin{eqnarray}
\hat{G}\hat{G}^{\dagger}
=
\left(
2\sum_{i=1}^{k} \ket{\psi_i}\bra{\psi_i}
-
\mathbb{I}
\right)
\left(
2\sum_{j=1}^{k} \ket{\psi_j}\bra{\psi_j}
-
\mathbb{I}
\right). \nonumber
\end{eqnarray}

\begin{eqnarray}
=
4\sum_{i,j=1}^{k}
\ket{\psi_i}\,
\langle \psi_i \vert \psi_j \rangle\,
\bra{\psi_j}
-
2\sum_{i=1}^{k} \ket{\psi_i}\bra{\psi_i}
-
2\sum_{j=1}^{k} \ket{\psi_j}\bra{\psi_j}
+
\mathbb{I}.
\end{eqnarray}

Since $\langle \psi_i \vert \psi_j \rangle = \delta_{ij}$, it follows that
\begin{eqnarray}
\hat{G}\hat{G}^{\dagger}
=
\mathbb{I}
=
\hat{G}^{\dagger}\hat{G}.
\end{eqnarray}
Therefore, $\hat{G}$ is unitary.

\end{proof}

\section{Spinorial quantum algorithms for classification and search}
{
In this section, based on the previous results, we propose three algorithms. The objective is to determine whether the items belong to a given class.

\begin{algorithm}[H]
	\caption{Quantum classification}
	\label{alg.:Algoritmo1}
	\begin{algorithmic}[1]
		\STATE \textbf{Input:} Quantum state
		\[
		\ket{\psi}
		=
		\cos \theta_{ij}\ket{\Gamma_j}+\sin \theta_{ij}\Gamma_{i}\ket{\Gamma_j}
		\]
		\IF{$\langle \psi | \Gamma_k | \psi \rangle = \delta_{ik}c_i $}
		\STATE The item belongs to Class $C_i$
		\ELSE
		\STATE The item does not belong to Class $C_{i}$
		\ENDIF
	\end{algorithmic}
\end{algorithm}

\begin{proposition}\label{sample2}
Consider Algorithm 1 applied to a finite dataset
\begin{eqnarray}
\mathcal D= \left\{|\psi_1\rangle,\dots,|\psi_M\rangle \right\},
\end{eqnarray}
where each quantum state is classified according to the expectation values
\begin{eqnarray}
\langle \psi | \Gamma_k | \psi \rangle= \delta_{ik}c_i, \qquad k=1,\dots,2n,
\end{eqnarray}
with \(c_i\neq0\). Suppose that for every state in the dataset there exists a margin $\eta>0$ such that
\begin{eqnarray}
|\langle \psi|\Gamma_i|\psi\rangle|\ge \eta
\end{eqnarray}
for the correct class and
\begin{eqnarray}
|\langle \psi|\Gamma_k|\psi\rangle|\le \eta/2, \qquad k\neq i.
\end{eqnarray}
For each observable $\Gamma_k$, the algorithm estimates the expectation value
$\langle \psi|\Gamma_k|\psi\rangle$ from $N$ independent measurements and classifies the input state according to the empirical averages.Then for any confidence parameter $\delta\in(0,1)$, there exists an integer $N_{\mathcal D}$ such that for every
$|\psi\rangle\in\mathcal D$, Algorithm 1 outputs the correct class with probability at least $1-\delta$.Furthermore, the sample complexity satisfies
\begin{eqnarray}
N=O\left(\frac{1}{\eta^2}\log\frac{n}{\delta}\right).
\end{eqnarray}
\end{proposition}
\begin{proof}
For each Clifford observable $\Gamma_k$, let
\begin{eqnarray}
X_r^{(k)}\in\{-1,+1\}, \qquad r=1,\dots,N,
\end{eqnarray}
denote the outcomes of $N$ independent measurements. Define the empirical estimator
\begin{eqnarray}
\hat{\mu}_k=\frac1N\sum_{r=1}^N X_r^{(k)}.
\end{eqnarray}
Since
\begin{eqnarray}
\mathbb E[\hat{\mu}_k]=\langle \psi|\Gamma_k|\psi\rangle,
\end{eqnarray}
and
\begin{eqnarray}
\langle \psi|\Gamma_k|\psi\rangle=\sin(2\theta_{ij})\delta_{ik},
\end{eqnarray}
Hoeffding's inequality  \cite{Handel} implies
\begin{eqnarray}
\Pr\left(|\hat{\mu}_k-\langle \psi|\Gamma_k|\psi\rangle|\ge\epsilon\right)\le2e^{-N\epsilon^2/2}.
\end{eqnarray}
Choosing
\begin{eqnarray}
\epsilon=\frac{\eta}{2},
\end{eqnarray}
where
\begin{eqnarray}
|\sin(2\theta_{ij})|\ge\eta>0,
\end{eqnarray}
the probability that the empirical estimator changes the classification is bounded by
\begin{eqnarray}
2e^{-N\eta^2/8}.
\end{eqnarray}
Since the algorithm evaluates \(2n\) observables, the union bound gives
\begin{eqnarray}
\Pr(\text{classification error})\le4n\,e^{-N\eta^2/8}.
\end{eqnarray}
Therefore, to guarantee total error probability at most \(\delta\), it suffices to choose \(N\) such that
\begin{eqnarray}
4n\,e^{-N\eta^2/8}\le\delta.
\end{eqnarray}
Solving for $N$,
\begin{eqnarray}
N\ge\frac{8}{\eta^2}\log\left(\frac{4n}{\delta}\right).
\end{eqnarray}
Hence,
\begin{eqnarray}
N=O\left(\frac{1}{\eta^2}\log\frac{n}{\delta}\right).
\end{eqnarray}
Finally, since the dataset $\mathcal D$ is finite, one may choose a uniform constant
$N_{\mathcal D}$ valid for all states in the dataset, completing the proof.

\end{proof}

\begin{algorithm}[H]
	\caption{Quantum classification}
	\label{alg.:Algoritmo2}
	\begin{algorithmic}[1]
		\STATE \textbf{Input:} Quantum state
		\[
		\ket{\psi}
		=
		\cos \theta_{ij}\ket{\Gamma_j}+\sin \theta_{ij}\Gamma_{i}\ket{\Gamma_j},
		\]
		
		\IF{$\langle \psi | \Gamma_j | \psi \rangle > 0$}
		\STATE The item belongs to Class A
		\ELSE
		\STATE The item belongs to Class B
		\ENDIF
	\end{algorithmic}
\end{algorithm}

This algorithm can be derived from a particular case of Definition 2 and Lemma 1 for two classes. The state $\ket{\psi}$ is then obtained as:
\begin{eqnarray}
	\ket{\psi}=\exp{(\theta\Gamma_{i}\Gamma_{j})}\ket{\Gamma_i}=\cos{\theta}\ket{\Gamma_{i}}+ \sin{\theta}\Gamma_{i}\Gamma_{j}\ket{\Gamma_j}
\end{eqnarray}
In fact, computing $\bra{\psi}O\ket{\psi}$, we obtain, for arbitrary $\theta$ and $O = \Gamma_j$,

\begin{equation}
	\begin{split}
		\bra{\psi}O\ket{\psi}
		&=
		\left(
		\cos\theta \bra{\alpha}
		+
		\sin\theta \bra{\beta}
		\right)
		\Gamma_j
		\left(
		\cos\theta \ket{\alpha}
		+
		\sin\theta \ket{\beta}
		\right)
		\\
		&=
		\cos^2\theta \, \innerproduct{\alpha}{\alpha}
		-
		\sin^2\theta \, \innerproduct{\beta}{\beta}
		\\
		&>0 .
	\end{split}
\end{equation}

On the other hand, for $O=-\Gamma_j$
\begin{equation}
	\bra{\psi}(-\Gamma_j)\ket{\psi} < 0 .
\end{equation}

The algorithm applies to quantum data, where each datum is represented by a quantum state $\ket{\psi}$. The classification is performed by measuring the expectation value of a Clifford generator, resulting in a decision rule in Hilbert space. In the quantum circuit model, both the state preparation and the measurement can be implemented using $O(k)$ elementary quantum gates, where $k$ denotes the number of qubits on which the Clifford generators act nontrivially. The advantage of the proposed classifier does not lie in an asymptotic speedup relative to classical algorithms operating on classical data, but rather in its inherent ability to process quantum data directly. In particular, a classical approach would require full or partial quantum state tomography to estimate the same expectation value, whereas the present algorithm accesses it operationally through direct measurement.

The following proposition analyzes the number of measurement repetitions (sample complexity) required for Algorithm 2 to achieve a given classification accuracy.

\begin{proposition} \label{sample}
	Consider Algorithm~2 applied to a finite dataset
	\begin{eqnarray}
		\mathcal D=\{\,|\psi_{\theta_1}\rangle,\dots,|\psi_{\theta_M}\rangle\,\},
	\end{eqnarray}
	where each input state is of the form
	\begin{eqnarray}
		|\psi_{\theta}\rangle
		=
		\cos\theta\,|\Gamma_i\rangle
		+
		\sin\theta\,\Gamma_i\Gamma_j|\Gamma_j\rangle,
	\end{eqnarray}
	and the classification observable is the Clifford generator \(O=\Gamma_j\). Suppose that no input state lies exactly on the decision boundary, that is,
	\begin{eqnarray}
		\cos(2\theta_\ell)\neq 0
		\quad \text{for all } \ell=1,\dots,M.
	\end{eqnarray}
	For each input state, the algorithm estimates the expectation value
	\(\langle \Gamma_j\rangle=\cos(2\theta)\) by repeating the measurement
	\(N\) times and classifies the state according to the sign of the empirical mean.
	Then, for any fixed confidence parameter \(\delta\in(0,1)\), there exists
	a finite integer \(N_{\mathcal D}\) such that, for every
	\(|\psi_{\theta_\ell}\rangle\in\mathcal D\), the algorithm outputs the
	correct class with probability at least \(1-\delta\).
	In particular, since \(\mathcal D\) is finite, the required number of
	measurements is a constant independent of the system size, and the sample
	complexity of Algorithm~1 on \(\mathcal D\) is \(O(1)\).
\end{proposition}

\begin{proof}
	Let \(|\psi_{\theta}\rangle\in\mathcal D\) be an inpute state.	A single measurement of the observable \(\Gamma_j\) produces a random variable
	\begin{eqnarray}
		X\in\{+1,-1\},
	\end{eqnarray}
	with expectation value
	\begin{eqnarray}
		\mathbb{E}[X]=\langle \psi_{\theta}|\Gamma_j|\psi_{\theta}\rangle=\cos(2\theta).
	\end{eqnarray}
	Let
	\begin{eqnarray}
		\bar X_N=\frac{1}{N}\sum_{k=1}^N X_k
	\end{eqnarray}
	be the empirical mean of \(N\) independent repetitions.
	Algorithm~1 assigns the class according to the sign of \(\bar X_N\). A misclassification occurs if and only if
	\begin{eqnarray}
		\operatorname{sign}(\bar X_N)\neq \operatorname{sign}(\cos(2\theta)).
	\end{eqnarray}
	As \(\cos(2\theta)\neq 0\) by hypothesis, this condition is equivalent to
	\begin{eqnarray}
		\bar X_N\,\cos(2\theta) < 0.
	\end{eqnarray}
	Suppose first that \(\cos(2\theta)>0\).
	Then a sign error implies \(\bar X_N\le 0\), and therefore
	\begin{eqnarray}
		\bar X_N - \cos(2\theta)\le -\cos(2\theta).
	\end{eqnarray}
	Taking absolute values yields
	\begin{eqnarray}
		\bigl|\bar X_N - \cos(2\theta)\bigr|
		\ge
		\cos(2\theta)
		=
		|\cos(2\theta)|.
	\end{eqnarray}
	If \(\cos(2\theta)<0\), a sign error implies \(\bar X_N\ge 0\), so that
	\begin{eqnarray}
		\bar X_N - \cos(2\theta)\ge -\cos(2\theta),
	\end{eqnarray}
	and again
	\begin{eqnarray}
		\bigl|\bar X_N - \cos(2\theta)\bigr|
		\ge
		|\cos(2\theta)|.
	\end{eqnarray}
	Thus, in all cases, an incorrect classification can occur only if
	\begin{eqnarray}
		\bigl|\bar X_N - \cos(2\theta)\bigr|
		\ge
		|\cos(2\theta)|.
	\end{eqnarray}
	The random variables \(X_k\) are independent and bounded in the interval
	\([-1,1]\).
	Therefore, by Hoeffding's inequality \cite{Handel}, for any \(\varepsilon>0\),
	\begin{eqnarray}
		\Pr\!\left(
		\bigl|\bar X_N - \cos(2\theta)\bigr|\ge \varepsilon
		\right)
		\le
		2\exp(-2N\varepsilon^2).
	\end{eqnarray}
	Since the dataset \(\mathcal D\) is finite and contains no state exactly on the decision boundary, we have
	\begin{eqnarray}
		g_{\min}
		:=
		\min_{|\psi_{\theta_\ell}\rangle\in\mathcal D}
		|\cos(2\theta_\ell)|
	\end{eqnarray}
	which exists and satisfies \(g_{\min}>0\). Choosing \(\varepsilon=g_{\min}\), we obtain the bound
	\begin{eqnarray}
		\Pr(\text{classification error})
		\le
		2\exp\!\left(-2N g_{\min}^2\right).
	\end{eqnarray}
	Let \(\delta\in(0,1)\) be a confidence level. Choosing
	\begin{eqnarray}
		N_{\mathcal D} \;\ge\;
		\frac{1}{2g_{\min}^2}
		\ln\!\left(\frac{2}{\delta}\right)
	\end{eqnarray}
guarantees that every state in $\mathcal{D}$ is classified correctly with probability at least \(1-\delta\). Since both \(g_{\min}\) and \(\delta\) depend only on the fixed dataset \(\mathcal D\) and the desired confidence level, and not on the Hilbert space dimension, the number of qubits, or the system size, the required number of measurement repetitions is constant once \(\mathcal D\) and \(\delta\) are fixed. Consequently, the number of samples required per input state does not scale with the size of the quantum system, and the sample complexity of Algorithm~1 on the dataset \(\mathcal D\) is \(O(1)\).
\end{proof}

Proposition \ref{sample} suggests an advantage of Algorithm 1 in the context of quantum data classification, since the task can be performed with constant sample complexity without reconstructing the quantum state, in contrast to classical approaches that typically require tomographic access to the data.

To validate Algorithm \ref{alg.:Algoritmo1}, the circuit was executed on the quantum processor \texttt{ibm\_torino} via the \textit{Runtime Primitive Sampler} of IBM Quantum. The experiment verified whether the encoding of states by Clifford operators preserves the theoretical structure under the effects of noise in a NISQ device. The results for several combinations of bivectors are presented in Figure \ref{fig:comparativo_hist_ACE_prob_alg1}.

\begin{figure}[H]
	
	\centering
	\subfloat{
		\includegraphics[scale=0.7]{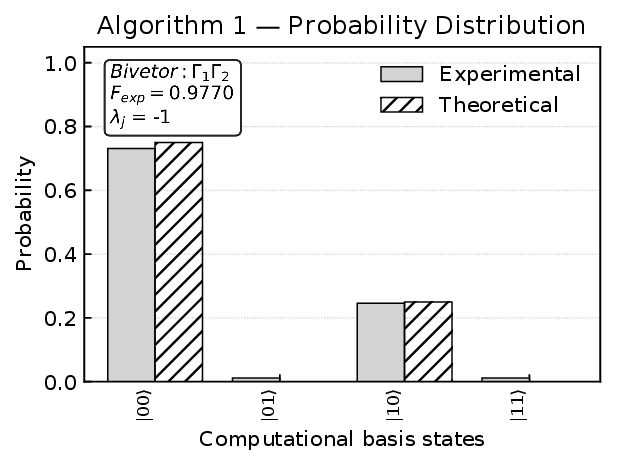}
	}
	\subfloat{
		\includegraphics[scale=0.7]{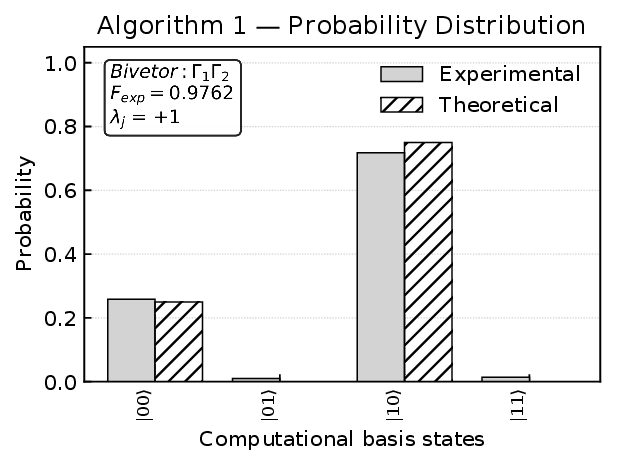}
	} \\
	\subfloat{
		\includegraphics[scale=0.7]{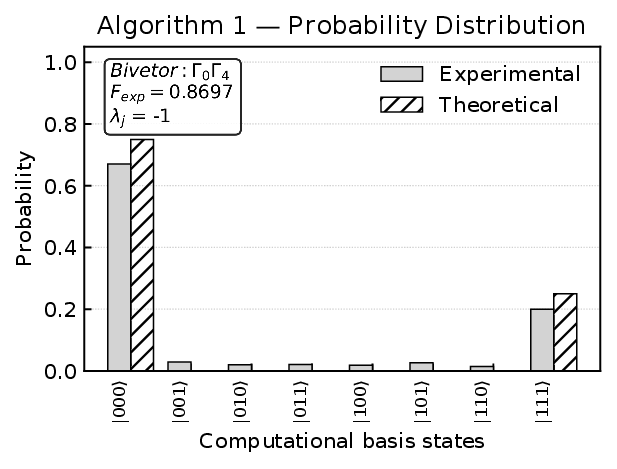}
	}
	\subfloat{
		\includegraphics[scale=0.7]{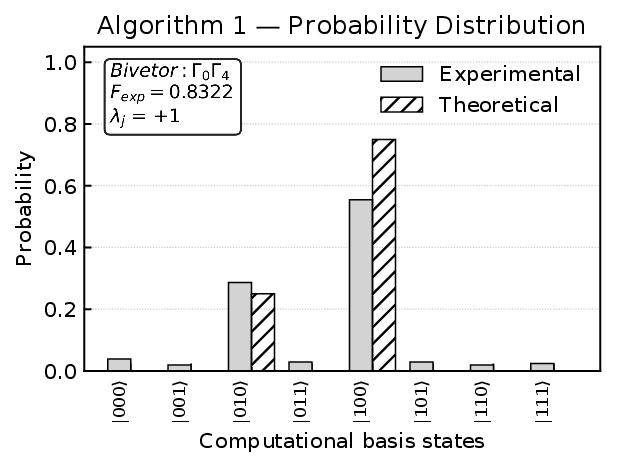}
	}\\
	\subfloat{
		\includegraphics[scale=0.7]{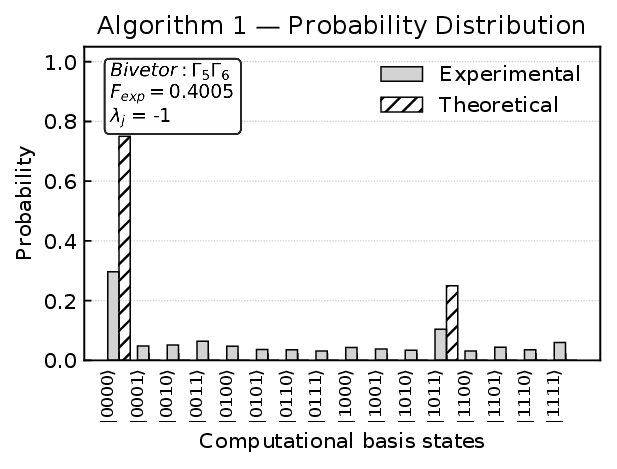}
	}
	\subfloat{
		\includegraphics[scale=0.7]{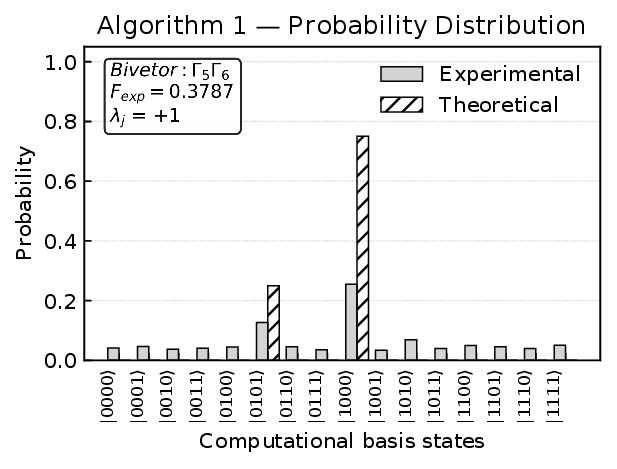}
	}\\
	\subfloat{
		\includegraphics[scale=0.7]{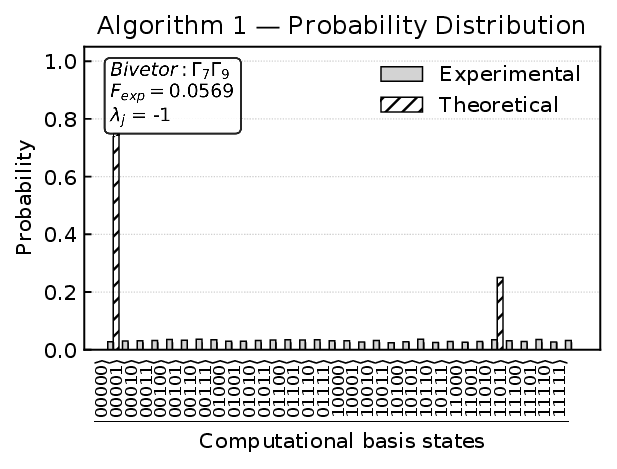}
	}
	\subfloat{
		\includegraphics[scale=0.7]{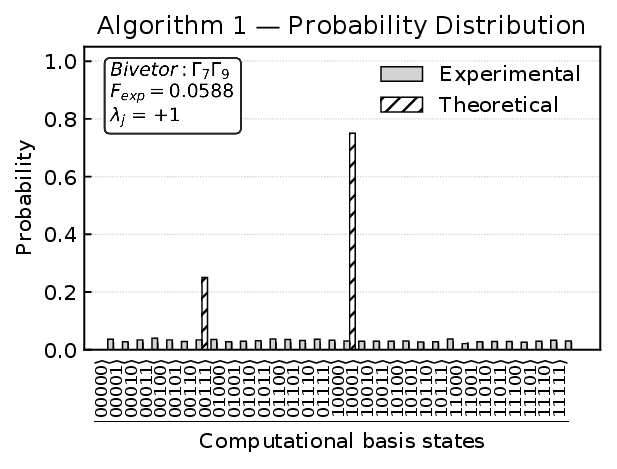}
	}\\
	\caption{Comparative probability distributions for Algorithm \ref{alg.:Algoritmo1} with 2, 3, 4, and 5 qubits. Each panel corresponds to a different classification operator. The histograms show the theoretical and experimental probabilities of the computational basis states after applying the operator O.}
	\label{fig:comparativo_hist_ACE_prob_alg1}
\end{figure}

The second algorithm is similar to Grover's algorithm, but here we consider a non-uniform initial distribution. In this context, we have a “bad” subspace, represented by $\ket{\alpha}$, and a “good” subspace, represented by $\ket{\beta}$, which contains the solution to the problem. The Grover operator is given by $(2\ket{\psi}\bra{\psi} - I)O$ with
\begin{eqnarray}
	\ket{\psi}
	=
	\cos(\theta)\ket{\alpha}
	+
	\sin(\theta)\ket{\beta}.
\end{eqnarray}

Here, unlike the standard Grover algorithm, we have $\theta$ corresponding to a non-uniform initial distribution, and the initial state is implemented through $R = \exp(\theta \Gamma_i \Gamma_j)$, as in Algorithm 1. Thus, for $O = \Gamma_j$, we obtain

\begin{equation}
	\begin{split}
		\ket{\psi'}
		&=
		\left(
		2\ket{\psi}\bra{\psi}
		-
		I
		\right)
		O \ket{\psi}
		\\
		&=
		\left(
		2\ket{\psi}\bra{\psi}
		-
		I
		\right)
		\left(
		\cos(\theta)\ket{\alpha}
		-
		\sin(\theta)\ket{\beta}
		\right)
		\\
		&=
		\cos(3\theta)\ket{\alpha}
		+
		\sin(3\theta)\ket{\beta}.
	\end{split}
\end{equation}

After $k$ iterations, we have
\begin{eqnarray}
	\vert \psi \rangle =\cos[(2k+1)\theta)] \vert \alpha \rangle + \sin[(2k+1)\theta)] \vert \beta \rangle
\end{eqnarray}

Therefore, the probability of obtaining the solution after $k$ iterations is
\begin{eqnarray}
	P_{sol}(k)= \vert \langle \beta \vert \psi_k \rangle \vert^{2}=\sin^{2}[(2k+1)\theta)]
\end{eqnarray}

so that $k = \left\lfloor \frac{\pi}{4\theta} - \frac{1}{2} \right\rceil$. We then obtain the following algorithm:

\begin{algorithm}[H]
	\caption{Quantum search with a non-uniform initial distribution}
	\label{alg:Algoritmo2}
	\begin{algorithmic}[1]
		\STATE \textbf{Input:} Quantum state
		\[
		\ket{\psi}
		=
		\cos\theta \ket{\alpha}
		+
		\sin\theta \ket{\beta}
		\]
		\STATE \textbf{Oracle:} Apply $\Gamma_j$
		\STATE \textbf{Diffusion:} Apply the operator
		\[
		D = (2\ket{\psi}\bra{\psi} - I)O
		\]
		\STATE \textbf{Iterations:} Repeat steps 2 and 3 $k = \left\lfloor \frac{\pi}{4\theta} - \frac{1}{2} \right\rceil$ times
		\STATE \textbf{Measurement:} Perform a measurement in the computational basis
	\end{algorithmic}
\end{algorithm}

In Figure \ref{fig:comparativo_hist_ACE_prob_alg2}, the four histograms obtained for Algorithm~2 illustrate the application of the proposed quantum search protocol to a non-uniform initial distribution based on spinorial representations. In all cases, the probability distributions are strongly concentrated around the marked state, indicating that the algorithm successfully amplifies the probability of the target element while suppressing the states that are not solutions.

\begin{figure}[H]
	\centering
	\subfloat{
		\includegraphics[scale=0.7]{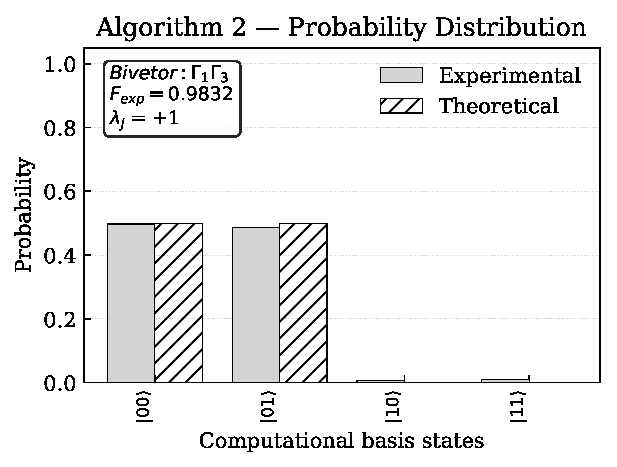}
	}
	\subfloat{
		\includegraphics[scale=0.7]{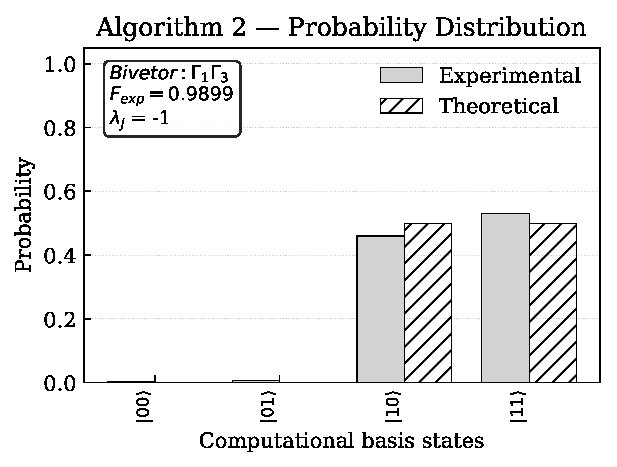}
	} \\
	\subfloat{
		\includegraphics[scale=0.7]{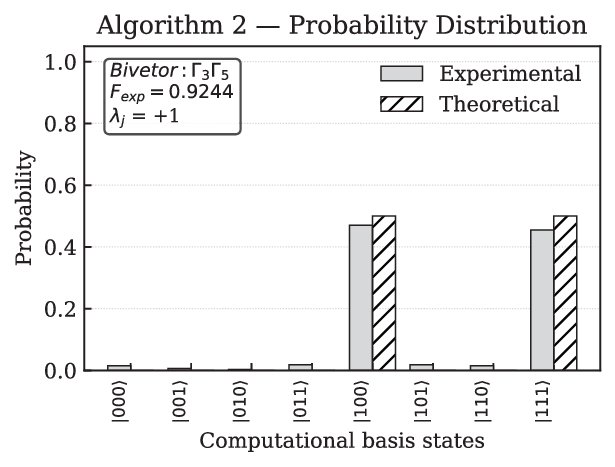}
	}
	\subfloat{
		\includegraphics[scale=0.7]{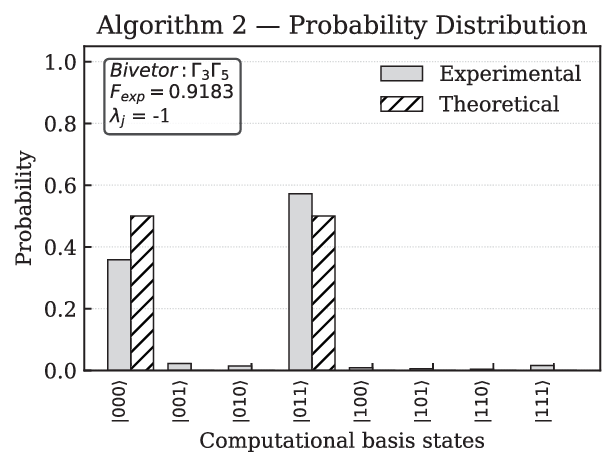}
	} \\
	\subfloat{
		\includegraphics[scale=0.7]{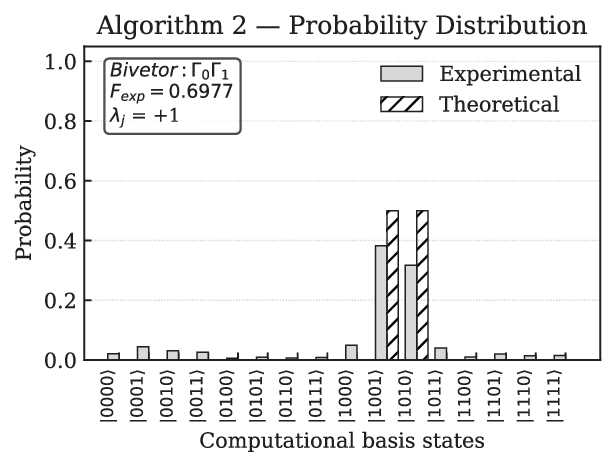}
	}
	\subfloat{
		\includegraphics[scale=0.7]{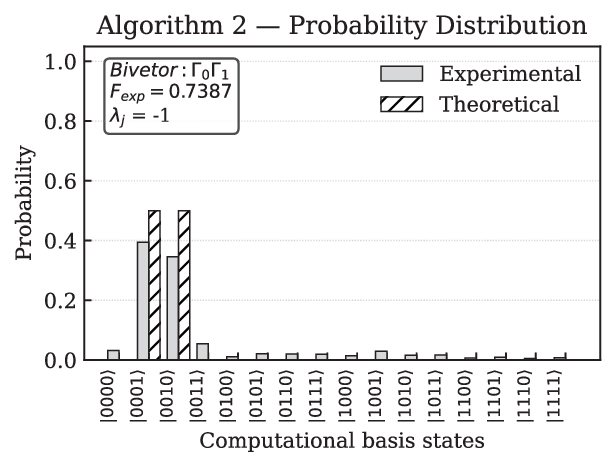}
	} \\
	\subfloat{
		\includegraphics[scale=0.7]{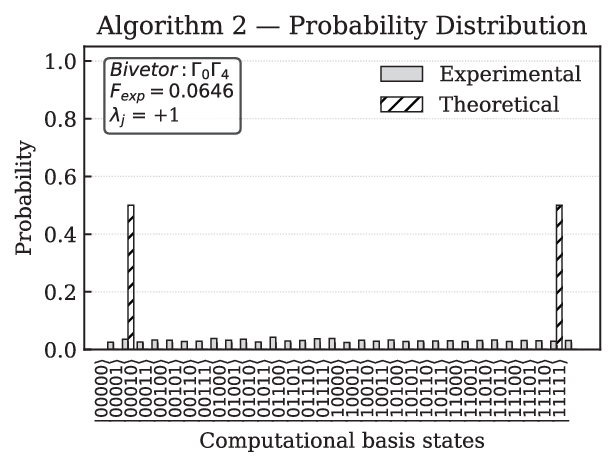}
	}
	\subfloat{
		\includegraphics[scale=0.7]{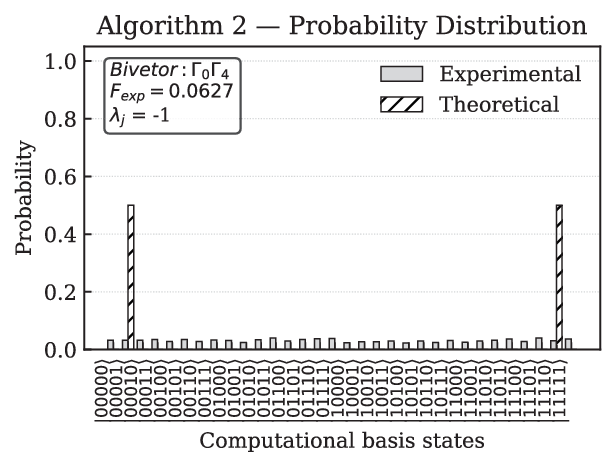}
	}
	
	\caption{Probability distributions of Algorithm \ref{alg:Algoritmo2} for two 2,3,4 and 5 qubits. One observes amplitude amplification in the solution subspace and suppression by destructive interference in the orthogonal complement. Solid bars represent experimental data and hatched bars represent theoretical predictions, showing the high fidelity of Clifford transformations in the decomposition of the state space.}
	\label{fig:comparativo_hist_ACE_prob_alg2}
\end{figure}

The decrease in experimental fidelity as the number of qubits increases, observed in Algorithms 1 and 2, is an intrinsic phenomenon of operation on NISQ devices and can be decomposed into structural and dynamical factors of the hardware. The implementation of the rotation operator requires a decomposition into universal gates whose circuit depth grows with the dimensionality of the spinorial representation. On \texttt{ibm\_torino}, this depth requires the insertion of additional SWAP gates to perform operations between non-adjacent qubits, which increases the total number of logical gates and consequently accumulates systematic and stochastic gate errors. Furthermore, the increase in total execution time relative to the coherence times of the device favors depolarization processes, which explains the dispersion observed in the probability distributions for larger-scale systems.

}

\section{Conclusions and Perspectives}

In this work, we developed an algebraic formulation based on Clifford algebras and their spinorial representations, which provides a unified language for the construction of quantum algorithms. Within this formalism, we proposed both classification and search algorithms whose functioning is entirely based on algebraic operations on spinorial states. One of the main advantages of this approach is that it avoids ad hoc constructions and instead exploits intrinsic algebraic properties, such as orthogonality and anticommutation relations, arising in Clifford algebras.

With respect to the classification algorithm, our proposal allows an arbitrary number of classes. The algorithm is conceptually simple and relies on the construction of mutually orthogonal quantum states, which are obtained from spinorial representations. This orthogonality ensures unambiguous discrimination between classes upon measurement, while the tensor structure guarantees systematic generalization to higher-dimensional systems.

Concerning quantum search, we analyzed a scenario involving non-uniform initial distribution in which the stored quantum states admit a Clifford encoding, also derived from spinorial representations. In this scenario, the oracle implementation becomes particularly simple, as it can be realized using a single generator of the Clifford algebra $Cl(2n)$. This represents a significant simplification when compared to more general oracle constructions, and highlights the computational advantages of encoding information directly in algebraic structures.

In a closely related approach, we may also show that a quantum search algorithm based on Hamiltonian simulation can be formulated within the same algebraic structure. By employing spinorial representations and Hamiltonians constructed from generators of the Clifford algebra, the search process can be interpreted as a controlled algebraic rotation in Hilbert space. In this context, the objective is to drive the system from an initial state $\ket{\psi}$ to a target state of the form $\Gamma_i\ket{\psi}$. Choosing the Hamiltonian as $H = -i\theta \Gamma_i\Gamma_j$, the time evolution yields $\ket{\psi'} = \exp(-itH)\ket{\psi} = \exp(it\theta \Gamma_i\Gamma_j)\left( \cos(\theta)\ket{\Gamma_j} - \sin(\theta)\Gamma_i\ket{\Gamma_j} \right) = \cos\left[\theta(1+t)\right]\ket{\Gamma_j} + \sin\left[\theta(1+t)\right]\Gamma_i\ket{\Gamma_j}$. An observation at the specific time $t = \left(\pi / 2\theta - 1\right)$ then results in the state $\Gamma_i\ket{\Gamma_j}$.

More broadly, these results suggest that Clifford algebras and spinorial methods constitute a promising formulation for the design of quantum algorithms, potentially opening new connections between quantum computation theory, Clifford algebras and spinors.

As perspectives, we intend to simulate the proposed classification and search algorithms for systems involving more than five qubits on better-calibrated quantum devices, with the goal of analyzing their scalability, robustness against noise, and potential advantages over standard quantum algorithms. 

\section{Appendix}
We review some basic facts about representations of Clifford algebras and spinorial representations of $SO(n)$. We follow the references \cite{Vaz, Sat, Gilbert}.
\begin{definition}
Let $V = \mathbb{R}^{2n}$ be a real vector space endowed with a symmetric bilinear inner product
\begin{eqnarray}
(\cdot,\cdot) : V \times V \to \mathbb{R}.
\end{eqnarray}
The tensor algebra $T(V)$ of $V$ is
\begin{eqnarray}
T(V) = \mathbb{R} \oplus V \oplus (V \otimes V) \oplus (V \otimes V \otimes V) \oplus \cdots,
\end{eqnarray}
with multiplication defined by the tensor product. Let $I$ be the two-sided ideal of $T(V)$ generated by the elements
\begin{eqnarray}
v \otimes v - (v,v)\, \mathbf{1}, \quad \forall v \in V.
\end{eqnarray}
The Clifford algebra associated with $(V, (\cdot,\cdot))$ is the quotient
\begin{eqnarray}
Cl(V, (\cdot,\cdot)) = T(V)/I.
\end{eqnarray}
For an orthonormal basis $\{e_i\}$ of $V$, the generators $\Gamma_i = \Gamma(e_i)$ satisfy
\begin{eqnarray}
\Gamma_i \Gamma_j + \Gamma_j \Gamma_i = 2 \, \delta_{ij} \, \mathbf{1}.
\end{eqnarray}
\end{definition}
\begin{definition}
Let $Cl^+(V, (\cdot,\cdot))$ be the even subalgebra of $Cl(V, (\cdot,\cdot))$, formed by products of an even number of unit vectors.
The Spin group is defined as
\begin{eqnarray}
\mathrm{Spin}(2n) = \{ R = v_1 v_2 \cdots v_{2k} \in Cl^+(V) \mid \|v_i\| = 1 \}.
\end{eqnarray}
The natural action of $\mathrm{Spin}(2n)$ on $V$ is
\begin{eqnarray}
\Pi(R)(v) = R v R^{-1}, \quad R \in \mathrm{Spin}(2n), \ v \in V,
\end{eqnarray}
defining a group homomorphism
\begin{eqnarray}
\Pi: \mathrm{Spin}(2n) \longrightarrow SO(2n),
\end{eqnarray}
with kernel $\{\pm 1\}$, so that $\mathrm{Spin}(2n)$ is a double covering of $SO(2n)$.
\end{definition}
\begin{definition}
Let $S$ be an irreducible left module over $Cl(V,(\cdot,\cdot))$, called the spinor space,
and let $GL(S)$ be the general linear group of invertible operators on $S$. The \textbf{spinorial representation} of $Spin(2n)$ is defined as the homomorphism
\begin{eqnarray}
\rho : \mathrm{Spin}(2n) \longrightarrow GL(S), \qquad \rho(R)\psi = R \psi,
\end{eqnarray}
where $R \in \mathrm{Spin}(2n)$ and $\psi \in S$.
Every element $R \in \mathrm{Spin}(2n)$ can be written as the exponential of a bivector:
\begin{eqnarray}
R = \exp\Bigl(\sum_{i<j} \theta_{ij} \, \Gamma_i \Gamma_j \Bigr),
\end{eqnarray}
and its action on spinors is
\begin{eqnarray}
\psi' = \rho(R)\psi = \exp\Bigl(\sum_{i<j} \theta_{ij} \, \Gamma_i \Gamma_j \Bigr)\psi.
\end{eqnarray}
\end{definition}
\begin{proposition}[Theorem 13.1 of reference \cite{Sat}]
Let $\Gamma_i$ satisfy the relations $\left\{\Gamma_i, \Gamma_j\right\} = 2\delta_{ij}$. Then the elements $\Gamma_{ij} = \frac{1}{2} \Gamma_i\Gamma_j$ generate the Lie algebra $\mathfrak{so}(2n)$.
\end{proposition}
\begin{proof}
We must compute the commutator $\left[\Gamma_{ij}, \Gamma_{rs}\right]$.
\begin{equation}
    \begin{split}
        \left[\Gamma_{ij}, \Gamma_{rs}\right] &= \frac{1}{4} \left[\Gamma_{i}\Gamma_j, \Gamma_{r}\Gamma_s\right] \\
        &= \frac{1}{4} \left(\Gamma_i\Gamma_j \Gamma_r\Gamma_s - \Gamma_r\Gamma_s\Gamma_i\Gamma_j\right).
    \end{split}
\end{equation}
Using the relation $\Gamma_s\Gamma_i = 2\delta_{is} - \Gamma_i\Gamma_s$, we obtain
\begin{equation}
    \Gamma_r\Gamma_s\Gamma_i\Gamma_j = 2\delta_{is}\Gamma_r\Gamma_j - 2\delta_{ir}\Gamma_s\Gamma_j + \Gamma_s\Gamma_i\Gamma_r\Gamma_j,
\end{equation}
and
\begin{equation}
    \Gamma_s\Gamma_i\Gamma_r\Gamma_j = 2\delta_{is}\Gamma_r\Gamma_j - 2\delta_{ir}\Gamma_s\Gamma_j + 2\delta_{sj}\Gamma_i\Gamma_r -  \Gamma_i\Gamma_j\Gamma_s\Gamma_r.
\end{equation}
Moreover,
\begin{equation}
    \Gamma_i\Gamma_s\Gamma_j\Gamma_r = 2\delta_{sj}\Gamma_i\Gamma_r - \Gamma_i\Gamma_j\Gamma_s\Gamma_r.
\end{equation}
Therefore,
\begin{equation}
    -\Gamma_r\Gamma_i\Gamma_s\Gamma_j = 2\delta_{is}\Gamma_r\Gamma_j - 2\delta_{ir}\Gamma_s\Gamma_j + 2\delta_{sj}\Gamma_i\Gamma_r -  \Gamma_i\Gamma_j\Gamma_s\Gamma_r,
\end{equation}
and
\begin{equation}
    -\Gamma_i\Gamma_j\Gamma_s\Gamma_r = 2\delta_{rj}\Gamma_i\Gamma_s + \Gamma_i\Gamma_j\Gamma_r\Gamma_s.
\end{equation}
Consequently,
\begin{equation}
    \begin{split}
        \left[\Gamma_{ij}, \Gamma_{rs}\right] &= \frac{1}{4} \left(
        \Gamma_i\Gamma_j\Gamma_r\Gamma_s - 2\delta_{is}\Gamma_r\Gamma_j + 2\delta_{ir}\Gamma_s\Gamma_j - 2\delta_{sj}\Gamma_i\Gamma_r - 2\delta_{rj}\Gamma_r\Gamma_s - \Gamma_i\Gamma_j\Gamma_r\Gamma_s
        \right) \\
        &= \delta_{is}\Gamma_{jr} + \delta_{ir}\Gamma_{sj} + \delta_{sj}\Gamma_{ri} + \delta_{rj}\Gamma_{si},
    \end{split}
\end{equation}
which corresponds to the Lie algebra $\mathfrak{so}(2n)$.
\end{proof}
Notice that there is an isomorphism $\mathfrak{spin}(2n)\simeq \mathfrak{so}(2n)$. We now prove the following proposition.
\begin{proposition}
Let $\Gamma_{i}$ and $\Gamma_{j}$ be generators of a Clifford algebra defined in the eq.(\ref{eq_def_gerador_clifford}). The operator
\begin{eqnarray}
R_{i,j}(\theta) = \exp\left( \theta \Gamma_{i}\Gamma_{j} \right)
\end{eqnarray}
can be written as
\begin{equation}
    R_{i, j}(\theta) = \cos(\theta)I + \sin(\theta)\Gamma_{i}\Gamma_{j}.
\end{equation}
\end{proposition}
\begin{proof}
From the power series expansion of the exponential, we have
\begin{equation}
    \begin{split}
        R_{i, j}(\theta) &= \exp\left( \theta \Gamma_{i}\Gamma_{j} \right)
        = \sum_{k=0}^{\infty}  \frac{\left( \theta \Gamma_{i}\Gamma_{j}\right)^k}{k!} \\
        &= \sum_{k=0}^{\infty} \frac{\theta^{2k}}{(2k)!} \left( \Gamma_{i}\Gamma_{j} \right)^{2k}
        + \sum_{k=0}^{\infty} \frac{\theta^{2k+1}}{(2k+1)!} \left( \Gamma_{i}\Gamma_{j} \right)^{2k+1}.
    \end{split}
\end{equation}
From Equation~\ref{eq_def_gerador_clifford}, we obtain
\begin{equation}
    \left( \Gamma_{i}\Gamma_{j} \right)^{2k} = (-1)^k I, \qquad
    \left(\Gamma_{i}\Gamma_{j} \right)^{2k+1} = (-1)^k \Gamma_{i}\Gamma_{j}.
\end{equation}
Therefore,
\begin{equation}
    R_{i,j}(\theta)
    = \sum_{k=0}^{\infty} \frac{(-1)^k \theta^{2k}}{(2k)!} I
    + \sum_{k=0}^{\infty} \frac{(-1)^k \theta^{2k+1}}{(2k+1)!} \Gamma_{i}\Gamma_{j}
    = \cos(\theta)I + \sin(\theta)\Gamma_{i}\Gamma_{j}.
\end{equation}
\end{proof}
The Pauli matrices are given by
\begin{eqnarray}
\sigma_1 =
\begin{pmatrix}
0 & 1 \\
1 & 0
\end{pmatrix},
\quad
\sigma_2 =
\begin{pmatrix}
0 & -i \\
i & 0
\end{pmatrix},
\quad
\sigma_3 =
\begin{pmatrix}
1 & 0 \\
0 & -1
\end{pmatrix}.
\label{eq:pauli_matrices}
\end{eqnarray}

\begin{theorem}
Let $\sigma_1,\sigma_2,\sigma_3$ be the Pauli matrices and $I$ the $2\times2$ identity matrix. For each $j=1,\dots,n$, define
\begin{eqnarray}
\Gamma_j = \sigma_3^{\otimes (j-1)}\otimes\sigma_1\otimes I^{\otimes (n-j)},
\qquad
\Gamma_{n+j} = \sigma_3^{\otimes (j-1)}\otimes\sigma_2\otimes I^{\otimes (n-j)}.
\end{eqnarray}
Then the $2n$ matrices $\{\Gamma_a\}_{a=1}^{2n}$ satisfy the Clifford relations
\begin{eqnarray}
\{\Gamma_a,\Gamma_b\} = \Gamma_a\Gamma_b + \Gamma_b\Gamma_a = 2\delta_{ab}\, I_{2^n},
\end{eqnarray}
that is, they provide a representation of the Clifford algebra $\mathrm{Cl}(2n)$ on $\mathbb{C}^{2^n}$.
\end{theorem}
\begin{proof}
We use the properties of the Pauli matrices
\begin{eqnarray}
\{\sigma_\alpha,\sigma_\beta\}=2\delta_{\alpha\beta}I,\qquad \sigma_\alpha^2=I,
\end{eqnarray}
and the tensor product property
\begin{eqnarray}
(A_1\otimes\cdots\otimes A_n)(B_1\otimes\cdots\otimes B_n)
=(A_1B_1)\otimes\cdots\otimes(A_nB_n).
\end{eqnarray}
We verify the Clifford relations by considering three cases.
\medskip

(i) Case $a=b$. \\
If $a=b$, then in $\Gamma_a^2$ one position contains either $\sigma_1^2$ or $\sigma_2^2$, and all other positions contain either $\sigma_3^2$ or $I^2$. Since all these squares equal the identity, it follows that $\Gamma_a^2=I_{2^n}$. Therefore,
\begin{eqnarray}
\{\Gamma_a,\Gamma_a\}=2I_{2^n}.
\end{eqnarray}

\medskip

(ii) Case $a\neq b$ belonging to the same block $j$ (for example, $a=\Gamma_j$, $b=\Gamma_{n+j}$). \\
In all tensor positions except $j$ the factors are identical; at position $j$ we have $\sigma_1$ and $\sigma_2$. Since
\begin{eqnarray}
\sigma_1\sigma_2=i\sigma_3,\qquad \sigma_2\sigma_1=-i\sigma_3,
\end{eqnarray}
we obtain
\begin{eqnarray}
\Gamma_j\Gamma_{n+j}=i\,T,\qquad \Gamma_{n+j}\Gamma_j=-i\,T,
\end{eqnarray}
for some common tensor factor $T$. Summing, we find $\Gamma_j\Gamma_{n+j}+\Gamma_{n+j}\Gamma_j=0$, hence $\{\Gamma_j,\Gamma_{n+j}\}=0$.

\medskip

(iii) Case $a$ and $b$ belonging to distinct blocks $j\neq k$. \\
Without loss of generality assume $j<k$. By construction, in all positions with index smaller than $j$ both tensor factors coincide; at position $j$ one tensor has $\sigma_{1}$ or $\sigma_{2}$ while the other has $\sigma_3$. Since $\sigma_3$ anticommutes with $\sigma_1$ and $\sigma_2$,
\begin{eqnarray}
\sigma_3\sigma_{\ell}=-\sigma_{\ell}\sigma_3\quad(\ell=1,2),
\end{eqnarray}
there is exactly one position where the local factors anticommute, while in all other positions they commute. Consequently, the products $\Gamma_a\Gamma_b$ and $\Gamma_b\Gamma_a$ differ by an overall minus sign, that is,
\begin{eqnarray}
\Gamma_a\Gamma_b = -\,\Gamma_b\Gamma_a,
\end{eqnarray}
and therefore $\{\Gamma_a,\Gamma_b\}=0$.
Combining cases (i)--(iii), we conclude that for any $a,b\in\{1,\dots,2n\}$,
\begin{eqnarray}
\{\Gamma_a,\Gamma_b\}=2\delta_{ab}\,I_{2^n},
\end{eqnarray}
which proves the theorem.
\end{proof}

\bigskip

\noindent\textbf{Remark.} To obtain a representation of the algebra $\mathrm{Cl}(2n+1)$, it suffices to add
\begin{eqnarray}
\Gamma_{2n+1}=\sigma_3^{\otimes n},
\end{eqnarray}
which satisfies $\Gamma_{2n+1}^2=I_{2^n}$ and anticommutes with all the $\Gamma_a$ defined above.
A multivector of grade $k$ in $\mathrm{Cl}(V,(\cdot,\cdot))$ is an element of the form
\begin{eqnarray}
A_k = v_1 v_2 \cdots v_k,
\end{eqnarray}
with $v_i \in V$, or a linear combination of such products. The Clifford algebra carries three canonical involutive operations, which we define explicitly.
The grade involution changes the sign of elements of odd degree:
\begin{eqnarray}
G(A_k) := (-1)^k A_k.
\end{eqnarray}
The reversion, denoted $R_v$, reverses the order of vectors in a product:
\begin{eqnarray}
R_v(v_1 v_2 \cdots v_k) := v_k v_{k-1} \cdots v_1.
\end{eqnarray}
When expressed back in the original order of vectors, it produces a sign factor
\begin{eqnarray}
R_v(A_k) = (-1)^{k(k-1)/2} \, A_k.
\end{eqnarray}
The Clifford conjugation, denoted by the superscript $C$, is defined as the composition of grade involution and reversion:
\begin{eqnarray}
A_k^C := R_v(G(A_k)) = (-1)^k \, R_v(A_k) = (-1)^{k(k+1)/2} A_k.
\end{eqnarray}
Let $\mathrm{Cl}(V,(\cdot,\cdot)) \otimes \mathbb{C}$ be complexification of Clifford algebra $\mathrm{Cl}(V,(\cdot,\cdot))$. After complexification, we also introduce the usual complex conjugation $\overline{\phantom{A}}$.  
The abstract adjoint is then defined by
\begin{eqnarray}
A^\dagger := \overline{A^C},
\end{eqnarray}
Assume now that $\dim V = 2n$ and fix an oriented orthonormal basis $(e_1,\ldots,e_{2n})$.  
Define the volume element
\begin{eqnarray}
\omega := e_1 e_2 \cdots e_{2n},
\end{eqnarray}
and the chirality operator
\begin{eqnarray}
\Gamma_* := i^n \, \omega.
\end{eqnarray}
Using the above involutive structure, one finds
\begin{eqnarray}
\Gamma_*^\dagger
= \overline{G(\Gamma_*)} 
= (-i)^n (-1)^{n(2n-1)} \, \omega
= i^n \, \omega
= \Gamma_*.
\end{eqnarray}
Therefore, the chirality operator is Hermitian in all even dimensions and satisfies
\begin{eqnarray}
\Gamma_*^2 = 1.
\end{eqnarray}
Consequently, the spinor space admits a canonical decomposition
\begin{eqnarray}
S = S_+ \oplus S_-, \qquad
S_\pm = \{ \psi \in S \mid \Gamma_* \psi = \pm \psi \},
\end{eqnarray}
with orthogonal projectors
\begin{eqnarray}
P_\pm := \frac{1}{2} (1 \pm \Gamma_*).
\end{eqnarray}
This provides representation-independent definition of chirality entirely within the structure of Clifford algebras.


\end{document}